\documentclass{svproc}

\usepackage{url}
\usepackage{hyperref}
\usepackage{multicol}
\usepackage{amsmath}
\usepackage{amssymb}
\usepackage{enumerate}  
\usepackage{subeqnarray}
\usepackage{graphicx}
\newtheorem{mechanism}{Mechanism}
\usepackage{subcaption}

\begin{document}
\mainmatter
\title{Repeated Auctions with Speculators: Arbitrage Incentives and Forks in DAOs}
\titlerunning{Repeated Auctions with Speculators}  
\author{Nicolas Eschenbaum, Nicolas J. Greber}
\institute{}

\maketitle

\begin{abstract} 
We analyze the vulnerability of decentralized autonomous organizations (DAOs) to speculative exploitation via their redemption mechanisms. Studying a game-theoretic model of repeated auctions for governance shares with speculators, we characterize the conditions under which---in equilibrium---an exploitative exit is guaranteed to occur, occurs in expectation, or never occurs. We evaluate four redemption mechanisms and extend our model to include atomic exits, time delays, and DAO spending strategies. Our results highlight an inherent tension in DAO design: mechanisms intended to protect members from majority attacks can inadvertently create opportunities for costly speculative exploitation. We highlight governance mechanisms that can be used to prevent speculation.
\end{abstract}

\section{Introduction}


Decentralized autonomous organizations (DAOs) are a novel governance structure enabled by blockchain technology that allows decentralized decision-making and management of collective resources. A major concern for such decentralized organizations is the risk of ``majority attacks,'' where a group holding the majority of voting rights decides to expropriate the collective resources.  To mitigate this, many DAOs allow members to exit before the decision takes effect, ensuring they receive their share of the joint assets upon departure. However, this mechanism can unintentionally attract speculators who join solely to exit with a share of the treasury.

In this paper, we study the incentives for speculators (or ``arbitrageurs'') to participate in repeated auctions for individual shares of the organization in order to profitably exploit this mechanism. We focus on the case of the Nouns DAO---one of the largest and most influential DAOs where proportional claims to the DAOs treasury are auctioned daily through unique NFTs.\footnote{In May 2023, the treasury of the DAO held about \$55 million, and the DAO has used its funds, for example, to fund a Super Bowl commercial with Bud Light in 2022 or name a rare species of frog in Ecuador. As of March 13, 2025 its treasury amounts to approximately \$6 million. See the Nouns website for the current treasury and auction at \href{https://nouns.wtf}{https://nouns.wtf}.} In September 2023, the exit from the DAO by a substantial share of its members led to a loss of (at the time) approximately \$27 million. While our analysis is centered on the Nouns DAO, the underlying economic mechanisms resemble a broader range of economic contexts, such as equity crowdfunding platforms with auction mechanisms or initial public offerings (IPOs) where participants similarly invest funds to acquire voting rights and proportional claims to organizational resources.

We study a game-theoretic model of a repeated English auction with two types of bidders: regular members (or ``Nouners'') who value their share (i.e., NFT) with some positive valuation and speculators (``Arbitrageurs'') whose valuation depends on the treasury size, conditions under which they may exit and redeem their share, and the resulting expected time until redemption.\footnote{In our model, we require a threshold number of shares to be held by speculators to force an exit.} We show that there are three types of equilibria depending on the likelihood of an exit---with an exit occurring with certainty, an exit occurring in expectation, or no exit with certainty, respectively---and characterize the conditions under which they arise. We then proceed to study four different exit mechanisms that change the share of the treasury received by members exiting and various extensions of the model that allow the organization to make exits less likely or prevent them altogether. In particular, we show how either commitment to a spending path or capping share values at the initial purchase price necessarily prevents participation by speculators.\footnote{Note that commitment is non-trivial to ensure for a decentralized organization and capping NFT values is technically challenging and adds properties to the NFTs that affect secondary markets.} We provide numerical simulations of our main findings that demonstrate how the different mechanisms and parameter values determine the equilibrium type.

Our work is closely related to the literature on auctions with speculators.  Bikhchandani and Huang in \cite{bukhchandani:huang} consider auctions where all bidders are pure intermediaries bidding solely for resale, as in treasury bill markets. Haile provides in \cite{haile} empirical evidence from timber auctions showing that the option to resell adds value for bidders and can inject a common-value element even in a private-value environment. Similar to our case, Garret and Tröger (see \cite{garratt-troger-2006}) show how a bidder with zero value who can profit by strategically bidding to win and then resell to a higher-value bidder introduces multiple equilibria to the game: besides the honest ``bid-your-value'' outcome, there exist speculative equilibria where the speculator wins at a low price and later sells to the real user.\footnote{Pagnozzi shows in \cite{pagnozzi} that, in fact, high-value bidders may sometimes prefer a speculator to win.} In contrast to this literature, we do not model an option to resell the share, but instead, we allow winners to redeem their share of the organization's assets in cash value. 

Our analysis is also related to the literature on initial public offerings (IPOs) and equity crowdfunding auctions. Many IPO participants are ``flippers'' who buy at the offering and resell quickly (``flipping'' the stock). Similarly, Lukkarinen and Schwienbacher (see \cite{lukkarinen:schwienbacher}) show that simply announcing a plan to enable secondary trading can significantly boost investor participation in equity crowdfunding.

This paper is structured as follows. Section \ref{sec:setup} develops the analytical framework. Section \ref{sec:prelim} characterizes the behavior of bidders in equilibrium. Section \ref{sec:main} provides our main results on the types of equilibria and conditions under which they arise. Section \ref{sec:ext} develops various extensions of our model and section \ref{sec:concl} concludes. All proofs are relegated to the appendix.

\section{Model Setup}\label{sec:setup}

Let $S_t$ denote the treasury stock at time $t \in 1,...,T$, with $T>1$, where $S_0 \geq 0$ is the treasury prior to the game commencing. The number of nouns at the start of the game is given by $N>0$ which are all held by nouners. At each time $t$ one new noun is auctioned off. Nouns serve as voting rights in the Nouns DAO.  We denote by $p_t$ the realized auction price at time $t$. The treasury stock therefore evolves over time according to $S_t = S_0 + \sum_{\tau = 1}^{\tau = t-1}p_\tau$.

Nouns can be redeemed for a share $\alpha_t$ of the treasury. We denote the expected value of one noun at time $t$ conditional on the history of the game $h_t$ and for an expected redemption period $t^\ast \geq t$ by $V^e_{t}(h_t)$, where we suppress the subscript of the expected redemption period $t^\ast$ for brevity.  We will refer to $V^e_{t}(h_t)$ as the redemption value. In the special case of $t^\ast=t$, we write the redemption value simply as $V^e_{t=t^\ast}(h_t) = V(h_t)$. 

All players discount future payoffs by the same discount factor $\delta \in (0,1)$. The expected redemption value is thus given by $V^e_{t}(h_t) = \delta^{t^\ast-t} \mathbb{E}[\alpha_{t} S_{t}(h_t)]$, where, as before, we suppress the subscript $t^\ast$ for brevity.

We assume that nouners value each noun according to a valuation $\tilde{v} \sim [0, \bar{v}]$, with $\bar{v}>0$ constant over time. Let $F$ denote the distribution function for $\tilde{v}$. We assume that $F$ has a density $f$ that is positive and continuous on $[0, \bar{v}]$ and is identically $0$ elsewhere. This represents the natural demand for nouns. Note that throughout, we will denote random variables with a tilde, e.g., $\tilde{v}$, and their realizations without, e.g., $v$. 

Arbitrageurs (or speculators), in turn, value the noun at the expected value of redemption  $V^e_t(h_t)$. In other words, we make a sharp distinction between nouners---who value nouns independent of their redemption value or any possible market value on a secondary market---and arbitrageurs, who only value them for the redemption value. Nouns can only be redeemed, however, after a fork occurs. We assume that only arbitrageurs want to fork and, when they do, will instantly redeem their nouns. For a fork to occur at time $t$, the number of nouns held by arbitrageurs, $A_t$, must exceed a constant forking threshold $\kappa \in (0,1)$, given by $\frac{A_t}{N+t} \geq \kappa$. At the start of the game, no arbitrageur holds a noun or $A_0 = 0$.

The timing of the model is as follows. At each time $t$, the auction commences with one arbitrageur and $n \geq 2$ nouners participating. The auction format is an English auction with a starting price of $0$. Following the auction, the realized auction price is added to the treasury, and a fork occurs if $A_t \geq \kappa (N+t)$. If it does, then all arbitrageurs fork and redeem their nouns, and the game ends. 

We model the auction format as an English auction, as this most closely resembles the auctions for nouns, in which bids can be submitted over a 24-hour period and all bids are visible at all times. Ending the game after the fork is WLOG, as all arbitrageurs leave the DAO and redeem their nouns. The continuation game is, therefore, simply a repetition of the game with a new starting treasury value $S_0$ and a number of nouns $N$. We assume that multiple nouners participate in order to guarantee that the expected auction price is always strictly positive.

The history $h_t$ contains treasury stocks and players' bids for all past periods. However, it is straightforward that the current treasury stock, number of nouns held by arbitrageurs at the beginning of the period, and sum of past prices contain all information required for players' strategies.\footnote{Depending on the specification of the share, the sum of past prices need not be part of the history. This is the case, for example, under the current pro-rata mechanic (see Section \ref{sec:main})} Let the sum of prices be denoted by $P_{t-1} = \sum_{\tau=1}^{t-1} p_\tau$. Then it is sufficient for the sequence of histories to be defined as $h_1 = \{S_0, A_0, P_0\}$, $h_2 = \{S_0, A_0, P_0, S_1, A_1, P_1\}$, etc. The set of histories is denoted by $H_t$.

Let $\tilde{b}_{i,t}(\cdot | h_t)$ denote the (possibly randomized) bid at which nouner $i \in 1,...,n$ plans to drop out at history $h_t$ as a function of her realized value $v_{i,t}$ and similarly let $\tilde{b}_{a,t}(\cdot|h_t)$ denote the arbitrageurs bid. We further denote by $\tilde{b}_{-i,t}$ the highest of the maximum bids of all players other than $i$ and by $\tilde{b}_{-a,t}$ the highest of the maximum bids of all nouners. For all bids $\tilde{b}_{i,t} \geq 0$ and $v_{i,t} \in [0,\bar{v}]$, nouner $i$'s expected payoff at time $t$ is then given by
\begin{equation}
u_{i,t}(b_{i,t}, v_{i,t}) = Prob \left[ \tilde{b}_{i,t} \geq \tilde{b}_{-i,t} \ | \ h_t \right] \left(v_{i,t}-\tilde{b}_{-i,t} \right).
\end{equation} 
Similarly, for all bids $\tilde{b}_{a,t} \geq 0$, the arbitrageur's expected payoff at time $t$ is given by
\begin{equation}
u_{a,t}(b_{a,t}, h_t) = Prob\left[ \tilde{b}_{a,t} > \tilde{b}_{-a,t} \ | \ h_t \right] \left( V^e_{t}(h_t) - \tilde{b}_{-a,t} \right).
\end{equation} Ties between nouners are broken at random\footnote{This assumption is innocuous in equilibrium, as the analysis below shows.}.

\section{Preliminaries}\label{sec:prelim}

We begin our analysis by studying the bidding behavior of players at any history $h_t$. 

\begin{lemma}\label{lem:bidding_strat}
(Bidding Strategies) In any equilibrium and at any $h_t$, 
\begin{enumerate}[i.]
    \item all nouners $i \in 1,...,n$ always bid up to their valuation of the noun, or $\tilde{b}_{i,t} = v_{i,t}$,
    \item all arbitrageurs always bid up to a cutoff value $\hat{b}_t(h_t)\geq \min \{ v_{n-1,t},V^e_t(h_t) \}$, or $\tilde{b}_{a,t}(h_t) = \hat{b}_t(h_t) \geq  \min \{ v_{n-1,t},V^e_t(h_t) \}$,
\end{enumerate} 
where $v_{n-1,t}$ denotes the second-highest value drawn among nounerns $i = 1,...,n$ at history $h_t$.
\end{lemma}
\begin{proof}
    in the appendix.
\end{proof}

Lemma \ref{lem:bidding_strat} shows that the equilibrium of the game consists of a pair of simple strategies for nouners and the arbitrageur. The first statement formalizes the standard result that it is optimal for players in an English auction to bid up to their true valuation. However, this result only directly applies to nouners, because an arbitrageur's valuation, the redemption value, is a function of her own bid and expected future bids and, thus, is determined in equilibrium. The second statement then shows that while pinning down arbitrageurs' bids in equilibrium is more complex, their optimal strategies still follow a simple structure of bidding up to a well-defined maximum bid $\hat{b}_t(h_t)$ that is at least as high as the lower of the arbitrageurs' redemption value and the second-highest nouner valuation. 

Since all nouners always bid up to their true valuation of the noun, we can now state the expected time to fork at a given period $t$ and history $h_t$ as a function of nouners valuations since the probability of arbitrageurs winning in expectation simply depends on the realizations of the highest nouner valuation $v_{n,t}$.
\begin{lemma}\label{lem:forking} (Forking): In any equilibrium and at any $h_t$, a necessary condition for $t^\ast$ to exist in expectation is given by
\begin{equation}
    A_t + \sum_{\tau = t}^{t^\ast} \left( \left[F(\hat{b}_{\tau})\right]^n \right) \geq \kappa (N+t^\ast).
\end{equation}
\end{lemma}
\begin{proof} 
in the appendix.
\end{proof}

Lemma \ref{lem:forking} shows that the expected time to fork (and the existence of a forking period at all) is driven by the distribution of nouners' valuations and the expected redemption value of arbitrageurs. The higher the redemption value, the more likely it is for the arbitrageur to win since arbitrageurs' will always at least outbid nouners whenever possible by Lemma 1. But an arbitrageur's win in expectation, therefore, only depends on the probability of their expected redemption value being above the highest nouner valuation; thus, they are simply a function of the distribution function of nouners' valuations.

Specifically, the probability of winning conditional on the expected redemption value and resulting optimal maximum bid, $F(\hat{b}_{t})$, must be greater than $\kappa$ for the game to move closer to a fork occurring in expectation. If it does, then, in expectation, arbitrageurs' wins accumulate faster than the 'cost' of a larger number of periods having been played and, thus, the required number of arbitrageur wins increasing (as a function of $\kappa$). But note that Lemma \ref{lem:forking} only provides a necessary condition. It does not guarantee that a fork will occur. The critical function of the condition in Lemma \ref{lem:forking} and the fact that it only guarantees a fork in expectation is illustrated in Figure \ref{fig:one}. Note that the displayed evolution of the expected number of arbitrageur wins, $E[A_t]$, is illustrative only.

\begin{figure}[htp]
    \centering
    \includegraphics[width=0.8\linewidth]{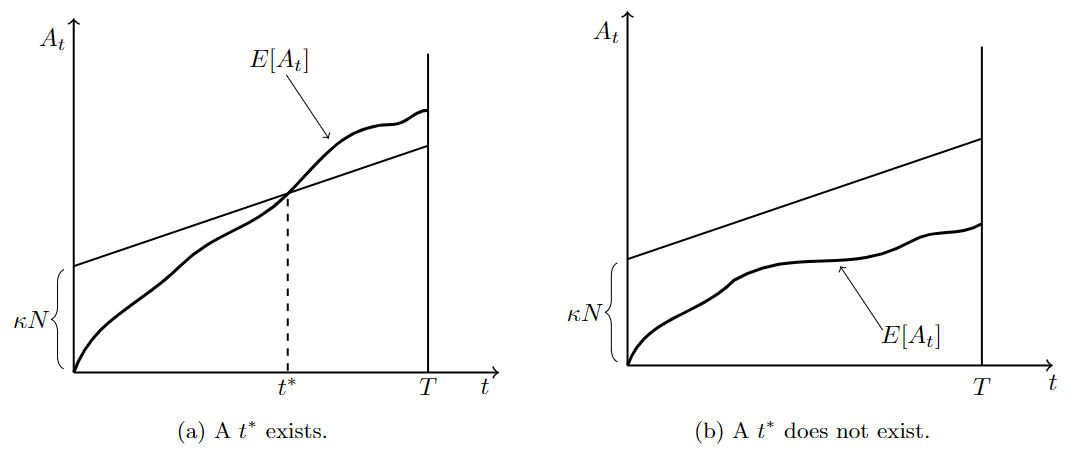}
    \caption{The relationship between the expected number of arbitrageur wins and the increasing number of nouns required to reach a fork.}
    \label{fig:one}
\end{figure}

Our final preliminary result now provides an explicit characterization of the optimal bidding behavior of arbitrageurs in equilibrium and shows how their behavior is determined by the effect an arbitrageur's bid has on the evolution of the redemption value. However, this only applies to interior solutions for the optimal maximum bid, denoted by $b_t^\ast(h_t)$. There can be equilibria in which an arbitrageur's maximum bid lies outside the range of nouners valuations. Intuitively, this would, for example, be the case if the initial treasury is so large that the expected value of a share of it is higher than any nouner is willing to pay for a noun. 

\begin{lemma}\label{lem:arbitrageur_bid} (Arbitrageur Bid): In any equilibrium and at any $h_t$, if the optimal maximal bid by an arbitrageur satisfies $b_t^\ast(h_t) \in [0,\bar{v}]$, then
\begin{equation}
b_t^\ast(h_t) = V^e_t(h_t) + \frac{1}{n} \frac{ F(\hat{b}_{a,t}(h_t))}{f(\hat{b}_{a,t}(h_t))} \frac{\partial V^e_t(h_t)}{\partial \hat{b}_{a,t}(h_t)}.
\end{equation}
\end{lemma}
\begin{proof}
    in the appendix.
\end{proof}

Lemma \ref{lem:arbitrageur_bid} shows that arbitrageurs optimize their bid relative to the redemption value. If the redemption value increases in their bid, then arbitrageurs may ``overbid''---that is, holding the redemption value fix for a given bid, they optimally bid more than this redemption value. If instead the redemption value decreases in their bid, then they would underbid. Whether arbitrageurs over- or underbid, therefore, depends on how the redemption value changes in their bid and the value of the hazard rate of nouners distribution at this bid level. As Lemma \ref{lem:bidding_strat} shows however, underbidding cannot arise in equilibrium since the optimal strategies of nouners are independent of arbitrageurs' strategies implying that while an arbitrageurs bid may affect the auction price (and, of course, may lead to him winning the auction), an increasing bid can never lead to a decrease in the auction price as this would require other players to react to the arbitrageurs' strategy.

\section{Main Results}\label{sec:main}

We now provide our main results on the equilibrium outcome of the game. Our first key finding studies the conditions under which a fork occurs in equilibrium.

\begin{proposition}\label{prop:equilibrium}(Equilibrium) 
The equilibrium of the game is one of three types, which are characterized by the following conditions:
\begin{enumerate}[Type I:]
    \item a fork is guaranteed to occur along the equilibrium path and the optimal maximum bid by arbitrageurs satisfies $b_t^\ast(h_t) \geq \bar{v}$ if $T \geq \frac{\kappa}{1-\kappa}N$ and $S_0 \geq \frac{\bar{v}}{\delta^{T-t} \alpha}$,
    \item a fork is expected to occur along the equilibrium path starting at history $h_t$ and the optimal maximum bid by arbitrageurs satisfies $b^\ast_t(h_t) \geq V^e_t(h_t)$ and $b_t^\ast(h_t) \in (0, \bar{v})$ if and only if 
    \begin{equation}
    \sum_{\tau=t}^{t^\ast} \left[F \left( b_t^\ast(h_t) \times g(\alpha, \tau, h_t) \right) \right]^n \geq \kappa(N+T) - A_t
    \end{equation}
    for any $t^\ast \leq T$,
    \item no fork is expected to occur along the equilibrium path starting at history $h_t$ and the optimal maximum bid by arbitrageurs satisfies $b^\ast_t(h_t) = 0$ if and only if 
    \begin{equation}
    \sum_{\tau=t}^{t^\ast} \left[F \left( b_t^\ast(h_t) \times g(\alpha, \tau, h_t) \right) \right]^n < \kappa(N+T) - A_t
    \end{equation}
    for all $t^\ast \leq T$,
\end{enumerate}
where $g(\cdot)$ is given by
$g(\alpha, \tau, h_t) = \frac{1}{\delta^{\tau-t}} \frac{\alpha_{\tau}}{\alpha_t} a_{\tau,t},
$ and \begin{subeqnarray}
            a_{\tau,t} &= \left[ \frac{\left( S_t + \sum_{\tau = t}^{t^*} \mathbb{E}[p_\tau] \right) + \frac{1}{n} \frac{F(b_t^*)}{f(b_t^*)} \sum_{\tau = t}^{t^*} \left( n F(b_t^*)^{n-1} (1 - F(b_t^*)) \right)}{ \left( S_t + \sum_{\tau = \tau}^{t^*} \mathbb{E}[p_\tau] \right) + \frac{1}{n} \frac{F(\mathbb{E}[b_{\tau}^*])}{f(\mathbb{E}[b_{\tau}^*])} \sum_{\tau = \tau}^{t^*} \left( n F(\mathbb{E}[b_{\tau}^\ast])^{n-1} (1 - F(\mathbb{E}[b_{\tau}^\ast])) \right)} \right] \\ &\geq 0.
\end{subeqnarray}
Furthermore, a sufficient condition for a Type II equilibrium is given by
\begin{equation}
    \sum_{\tau=t}^{t^\ast} \left[F \left( V^e_t(h_t) \times g(\alpha,\tau,h_t) \right) \right]^n \geq \kappa(N+T) - A_t \ \textnormal{for any} \ t^\ast \leq T.
\end{equation}
\end{proposition}
\begin{proof}
    in the appendix.
\end{proof}

Proposition \ref{prop:equilibrium} characterizes the three possible types of equilibria and provides conditions under which the game will end in each type of equilibrium. Intuitively, the result of the game depends on whether arbitrageurs' optimal bids lie within the interval of the distribution of nouner valuations. If they do and are sufficiently high, then the accumulating expected wins by arbitrageurs can be high enough for arbitrageurs to expect a fork to occur and be willing to bid a positive amount (Type~II). If they are not sufficiently high, then in expectation, no fork will occur before the game ends, so that arbitrageurs are not willing to bid any positive amount (Type~III). Finally, if they lie outside the bounds of the distribution, then arbitrageurs will always win, and as long as the game lasts long enough for them to surpass the forking threshold, a fork is guaranteed (Type I). However, note that our conditions for Type I only characterize the case in which the initial treasury stock $S_0$ is so large that a Type I equilibrium arises. There can be intermediate parameter values that also result in Type I equilibria being played, which are not captured by our conditions.

We also note that our conditions are not precise restrictions on the model primitives that must be fulfilled for the optimal bid by arbitrageurs to result in the respective type of equilibrium, but rather are expressed in terms of the optimal arbitrageur bid and thus the redemption value (which is a function of the model primitives), because the expected time to fork based on the sequence of win probabilities cannot be solved analytically. However, the conditions show that in order to characterize the equilibrium outcome of the model, it is sufficient to check whether, at the beginning of the game, arbitrageurs expect a fork to occur before the game ends. Moreover, to specifically check for a Type III equilibrium, the sufficient condition we provide, which is only a function of the redemption value rather than the full optimal bid by arbitrageurs, can be calculated instead.

Figure \ref{fig:two} illustrates the equilibrium mechanics of the game. The expected optimal bids in future periods are increasing in line with the discount factor, starting with the optimal bid in period 1 (left panel), resulting in an increasing probability of winning (right panel). The total accumulated win probabilities of arbitrageurs (shaded area) represent the expected number of arbitrageur wins, and if this is sufficiently high, then a fork is expected to occur. In this case, the positive bids by arbitrageurs are indeed equilibrium play, and we are in a Type II equilibrium. The same logic holds true starting at any other period $t$ instead of period 1 or for a forking period that is expected to occur before time $T$. In a Type III equilibrium, we instead observe an optimal arbitrageur bid of zero and the corresponding probability of winning of zero, while in a Type I equilibrium, the optimal bid lies above the bound of nouners valuations $\bar{v}$ (and so does each future optimal bid) and thus the winning probability is always equal to one. Note that we display the bids and accumulating probability with continuous functions while---in our model---time is discrete. 
\begin{figure}
    \centering
    \includegraphics[width=0.8\linewidth]{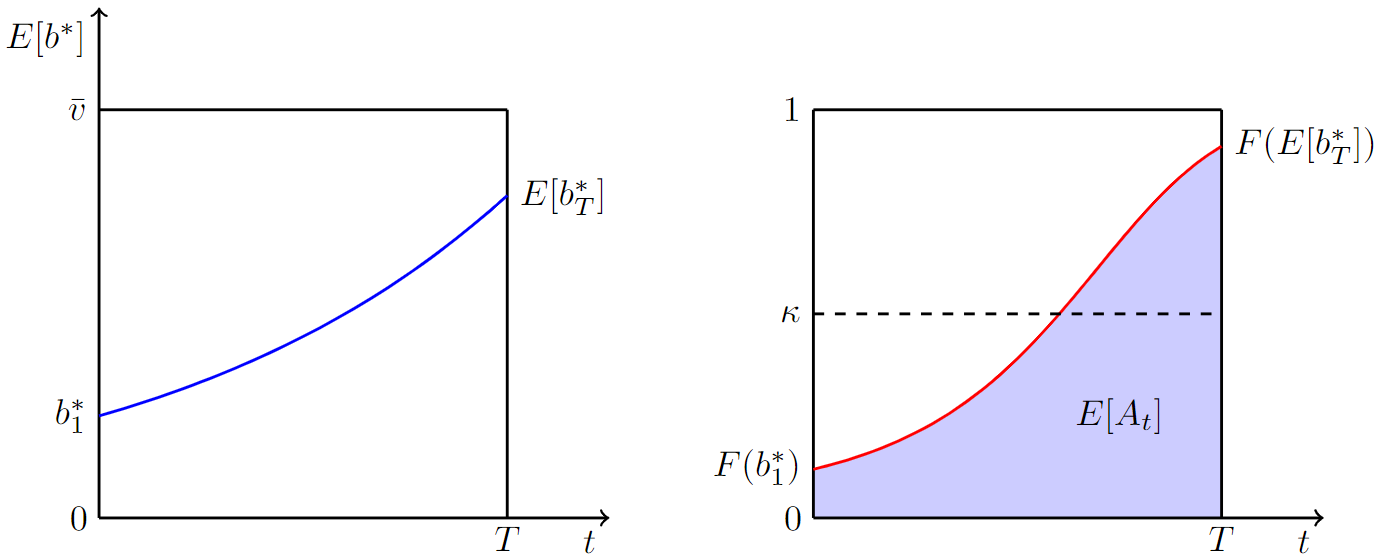}
    \caption{Example of expected optimal future bids at period 1 and resulting expected accumulating arbitrageur wins over time for the case of nouners valuations being distributed according to a normal distribution.}
    \label{fig:two}
\end{figure}
Lastly, consider that the conditions stated in Proposition \ref{prop:equilibrium} hold for various possible assumptions relevant for the redemption value. In particular, we have not specified the share arbitrageurs receive when redeeming a noun. The conditions stated in Proposition \ref{prop:equilibrium} only assume that the share arbitrageurs receive is not decreasing in their bid, that is, $\partial \alpha_{t^\ast}/\partial b_t^\ast(h_t) \geq 0$. We now proceed to numerically solve the model under different assumptions on the share arbitrageurs receive when redeeming a noun. We begin by imposing the current mechanism in the DAO on the model, specifically:
\begin{mechanism}\label{mech:one}
    (Pro-rata share): $\alpha_t = \frac{1}{N+t}$.
\end{mechanism}

When redeeming a noun, players receive the pro rata share of the entire treasury. Table \ref{tab:one} in the appendix shows for the case of Mechanism \ref{mech:one} how different values for the forking threshold $\kappa$ and the initial treasury $S_0$ affect the outcome of the game for one particular set of parameter values and distribution function of nouners. As expected, lower forking thresholds or higher initial treasuries reduce the redemption value and, hence, the probability of a fork and expected time to fork. 

In this setting, all three types of equilibria are possible, and even with an initial treasury of zero, a fork can still occur. This is because it can be profitable for arbitrageurs to play relatively low but positive maximum bids, resulting in nouners winning in expectation for a number of rounds and paying their prices into the treasury. Once arbitrageurs are able to force a fork, the nouners' paid prices result in a net gain for arbitrageurs. However, the forking threshold ($\kappa$) must be quite low for this to be an equilibrium strategy. 

The simulation runs for Table \ref{tab:one} in the appendix were conducted with the following parameters: the maximum number of periods was set to $T = 30$, with the number of initial nouners $N = 10$ and a discount factor $\delta = 0.95$. The number of nouners participating in the auctions was set at $n = 2$ and the upper bound of the uniform distribution at $\bar{v} = 1$.

The second assumption on the share arbitrageurs receive when redeeming the noun that we study  is the following:
\begin{mechanism}\label{mech:two}
    (Pro-rata share with tax): $\alpha_t = c \times \frac{1}{N+t}$ with $c\in (0,1)$.
\end{mechanism} 

In this case, arbitrageurs only receive a fraction of their pro-rata share. The remaining sum either stays in the DAOs treasury or is burned. It is straightforward that all results from the pro-rata mechanism continue to apply, as it is simply the special case of Mechanism \ref{mech:two} for the case that $c=1$. If the remaining sum stays in the DAOs treasury, our analysis is exactly identical to before. Whereas if the remaining sum is burnt, then all equilibrium mechanics continue to remain the same, but the starting initial treasury following the fork is lower.

Tables \ref{tab:two} and \ref{tab:three} in the appendix document the numerical results for Mechanism \ref{mech:two}. The effect of setting a tax is similar to the effect of a higher $\kappa$ or lower initial treasury: it reduces the redemption value and, thereby, the willingness to pay by arbitrageurs. In our simulations, the tax parameter ($c$) was set to $0.75$ and $0.5$, equivalent to a 25\% and 50\% tax, respectively. 

As the tax rate increases, equilibrium conditions are met at lower initial treasury levels or higher forking thresholds. Lowering $c$ makes forking less likely and extends the time to fork. This effect is less pronounced than varying the initial treasury because even a smaller share of a positive treasury remains valuable, whereas changing the initial treasury size more directly impacts the profitability of bidding for arbitrageurs.

The third forking mechanism that we consider is the following, where we denote by $j \leq t$ the time period in which a noun was bought:

\begin{mechanism}\label{mech:three}
    (Contribution-based share 1): $\alpha_{j,t}(h_t) = \frac{p_j}{\sum_{\tau=1}^{t} p_{\tau}}$. 
\end{mechanism}

In contrast to the previous two mechanisms, under Mechanism \ref{mech:three} the share received after a fork does not depend on the number of nouns. Instead, it is a function of the price paid for the noun or the contribution the owner of the noun has made to the size of the treasury. 

Tables \ref{tab:four} and \ref{tab:five} document the simulation results for Mechanism \ref{mech:three}. Compared to the pro-rata mechanism, Mechanism \ref{mech:three} leads to more "extreme" equilibria, that is, Type II equilibria become very unlikely and in almost all scenarios, arbitrageurs either are guaranteed to win or do not participate in the auction (i.e., Type I and Type III equilibria). The reason for this finding is that under Mechanism \ref{mech:three}, arbitrageurs are guaranteed to always obtain at least the price they paid. So they may as well bid high enough to guarantee themselves a win, whenever they can obtain a positive payoff by forking. However, they will only obtain this price at some future forking period $t^\ast$. Thus, due to discounting, they will value this return slightly less than the price they paid and Type II equilibria will become possible again. But this can only be the case under very specific parameter combinations, making Type II equilibria very rare.

The final mechanism that we consider is the following:

\begin{mechanism}\label{mech:four} (Contribution-based share 2): $\alpha_{j,t}(h_t) = \min \left\{ \frac{p_j}{S_{t}}, \frac{p_j}{\sum_{\tau=1}^{t} p_{\tau}} \right\}$.
\end{mechanism}

Compared to Mechanism \ref{mech:three}, Mechanism \ref{mech:four} puts a cap on the share arbitrageurs may receive. Just like in Mechanism \ref{mech:three}, arbitrageurs receive a share proportional to the auction price they paid relative to the sum of auction prices, $p_j/P_t$. But whenever this share exceeds the ratio of their paid auction price to the total value of the treasury, they obtain this second share instead. Clearly, this would be the case when the initial treasury is non-empty.

Figure \ref{fig:three} in the appendix illustrates the logic of this mechanism. The redemption value for arbitrageurs does not continuously grow with the size of the treasury, instead it is capped at the level of their paid auction price relative to the overall size of the treasury. 
The consequence of this mechanism is stated formally in the next result.

\begin{corollary}\label{cor:mech4}
    (Mechanism 4): Assume Mechanism \ref{mech:four}. Then, no Type I or Type II equilibria can arise and no forking occurs.
\end{corollary} 
\begin{proof}
    in the appendix.
\end{proof}

\section{Extensions}\label{sec:ext}
\subsection{Atomic Exit}

Instead of players being only able to fork once a sufficient number of noun holders choose to fork simultaneously, players can be allowed to exit individually at any point in time (an 'atomic exit'). Our model can capture this mechanism by placing two assumptions on the framework. First, we assume that $\kappa = 0$. As a consequence, the forking condition in Lemma \ref{lem:forking} is no longer required, as it will always be fulfilled once an arbitrageur wins their auction. Second, the redemption value for arbitrageurs simplifies to 
$$
V^e_t(h_t) = V(h_t) = \alpha_t (S_{t-1} + \mathbb{E}[p_t]),
$$ as the arbitrageur can fork after winning at time $t$ immediately, and thus, her value for the noun no longer depends on future auction prices or expected future arbitrageur wins. We can then state the following result, where we also characterize the development of the treasury as forks occur (i.e. if the game does not end after the first abers' exit).

\begin{proposition}\label{prop:atomic_exit}(Atomic exit) Assume that $\kappa=0$ so that players may fork and redeem their noun at any time $t$. Then, in any equilibrium and at any history $h_t$, a fork (atomic exit) occurs in expectation if
\begin{subeqnarray}
    b^\ast(h_t) \in [0, \bar{v}] \ \ \text{and} \ \ F\left( V(h_t) \right) > F(\mathbb{E}[v_{n,t}]) ,
\end{subeqnarray}
and the expected value of the treasury follows a mean-reverting process denoted by $\tilde{S_t}$ and defined by
\begin{subeqnarray}
    \tilde{S_t} \sim \mathcal{D}(\mu_t, \underline{S}, \bar{S}),
\end{subeqnarray}
where
\begin{equation}
\mu_t = (1 - \alpha_t F(b^\ast(h_t))^n)(\mu_t + \mathbb{E}[p_t]), \quad \underline{S} > 0, \quad \bar{S} < \frac{1}{\alpha_t} \bar{v}.
\end{equation}
\end{proposition} 
\begin{proof}
    in the appendix.
\end{proof}

As Proposition \ref{prop:atomic_exit} shows, in a steady state, the treasury will follow a treasury path defined by the exit mechanism and the expected value of the second highest nouner. Any treasury values that exceed this path will immediately result (in expectation) in an arbitrageur buying a noun to exit with profit. The consequence of this change to the forking mechanic in the DAO is particularly stark under Mechanism \ref{mech:three} and \ref{mech:four}, as the following statement shows.

\begin{corollary}(Arbitrage under atomic exit)
    At any time $t$ and history $h_t$, in any equilibrium:
    \begin{itemize}
        \item under Mechanism 3, a fork occurs with certainty whenever $S_{t-1} > P_{t-1}$.
        \item under Mechanism 4, no fork occurs.
    \end{itemize}
\end{corollary} 
\begin{proof}
    in the appendix.
\end{proof}

With the atomic exit mechanic, arbitrageurs no longer need to consider future expected arbitrageur wins. As a consequence, because they are guaranteed under Mechanism \ref{mech:three} and \ref{mech:four} to always obtain at least the price they paid in the auction, they will simply pay enough to be guaranteed to win the auction and fork immediately, whenever it yields a strictly positive payoff. This is the case under Mechanism \ref{mech:three} and a positive initial treasury. Under Mechanism \ref{mech:four}, however, the arbitrageurs' returns are capped and can never be strictly positive, and thus, we obtain the finding from Corollary \ref{cor:mech4} again that arbitrageurs will refrain from participating in the auction altogether.

\subsection{Time Delay}

When players choose to fork and redeem their noun, the share of funds they receive can be delayed. The implementation of a vesting period significantly influences the willingness of arbitrageurs to participate. As the delay in receiving funds increases, the willingness of arbitrageurs to participate decreases. This vesting period can be represented with a parameter $\Delta \in 1,2,...,T-1$ that represents the number of periods (i.e., days) during which the funds are locked.
The redemption value then becomes:
\begin{equation}
    V^e_t(h_t) = \delta^{ t^\ast-t + \Delta} \alpha_{t^\ast} \left( S_{t-1} + \sum_{\tau = t}^{t^\ast} \mathbb{E}[p_\tau] \right).
\end{equation}

The main consequence of this change to the forking mechanism is that it reduces the immediate financial incentive for arbitrageurs due to the delayed access to funds.

It is straightforward that the results from Proposition \ref{prop:equilibrium} carry over to this scenario, as stated formally in the next result.

\begin{corollary}\label{cor:vesting}(Vesting):
Assume funds received from redeeming a noun are ves-ted for $\Delta$ periods. Then, the conditions for a Type II and Type III equilibrium from Proposition 1 continue to apply. The condition for a Type I equilibrium becomes:
\begin{enumerate}[Type I:]
    \item a fork is guaranteed to occur along the equilibrium path and the optimal maximum bid by arbitrageurs satisfies $b_t^\ast(h_t) \geq \bar{v}$ if $T \geq \frac{\kappa}{1-\kappa}N$ and $S_0 \geq \frac{\bar{v}}{\delta^{T-t+\Delta} \alpha}$.
\end{enumerate}
\end{corollary} 

Intuitively, because our conditions for a Type II and Type III equilibrium are a function of the optimal bid and thus the expected redemption value, their formal definition does not change, but the set of parameters for which they apply changes. Similarly, our condition for a Type I equilibrium now requires a larger initial treasury to guarantee a fork.

\subsection{Treasury Spending}

In our baseline model, all outflows from the treasury arise from players forking and ragequitting the new DAO. In practice, funds from the treasury are regularly spent on proposals. Intuitively, introducing spending into the model will make forks less likely as the treasury shrinks and incentives to participate for arbitrageurs are reduced. Fully modeling the voting mechanism that decides on spending is beyond the scope of this analysis. However, we now introduce a reduced-form spending function into our setting and analyze its effects.

Let $z_t(h_t)$ denote treasury spending at time $t$ and history $h_t$. We will assume that the level of treasury spending is only a function of the current treasury stock at the beginning of the period so that 
\begin{equation}
z_t(h_t) = z_t(S_{t-1}).
\end{equation} The treasury stock, therefore, evolves over time according to
\begin{equation}
S_t = S_{t-1} + p_t - z_t(S_{t-1}).
\end{equation} We further assume that the timing in each period is such that the auction price is added to the treasury first, then treasury spending occurs, and finally, a fork can arise. It is straightforward then that all our main results continue to hold but that the redemption value for arbitrageurs now needs to be adjusted for the expected level of treasury spending or
\begin{equation}
V^e_t(h_t) = \delta^{t^\ast-t} \alpha_{t^\ast} \left( S_{t-1} + \sum_{\tau = t}^{t^\ast} \left( \mathbb{E}[p_\tau] - \mathbb{E}[ z_\tau(S_{\tau-1}) ] \right) \right).
\end{equation} It is immediately clear that depending on the expected level of spending, forks can be prevented entirely by, for example, committing to spend the entire treasury in each period. More generally, if the DAO were to commit to a spending path over time, the same result can be obtained despite not spending everything at once. Specifically, suppose that the DAO commits to a declining spending path that takes the form of an exponential decay, or
\begin{equation}
z_t(S_{t-1}) = k S_{t-1} e^{-\lambda(t-1)},
\end{equation} where $0 < k < 1$ is a constant determining the fraction of the treasury spent in each period and $\lambda > 0$ is the decay rate. Then, we can state the following result.

\begin{proposition}\label{prop:treasury_spending} (Treasury Spending): For any given value of $S_0$, there exists a set of critical parameters $\hat{k}$, $\hat{\lambda}$, such that no Type I or Type II equilibrium can arise.
\end{proposition}
\begin{proof}
    in the appendix.
\end{proof}

Proposition \ref{prop:treasury_spending} shows that by choosing an appropriate spending path, incentives for arbitrageurs can be sufficiently reduced to prevent forking. Intuitively, by spending sufficiently quickly, arbitrageurs' payoffs by the time they have accumulated enough wins to force a fork are reduced far enough so that their optimal bids become low enough to slow down the speed at which they accumulate wins in expectation to ensure the forking condition cannot be satisfied. As a consequence, even though they could obtain a positive payoff if a fork would arise, they can never get there, so it no longer pays to participate in the auction, and only Type III equilibria remain. A similar set of parameters could also be determined at which only Type I equilibria are excluded.

\section{Conclusion}\label{sec:concl}

This paper studies the incentives for speculators to participate in repeated auctions for shares of the assets of a decentralized organization. Redemption mechanisms that allow the winner of the auction to later redeem their share for cash value incentivize speculators to enter bids with the sole objective of exploiting the redemption mechanism. We build a game-theoretic model inspired by the Nouns DAO and characterize equilibrium bidding behavior and conditions under which successful arbitrage occurs.

We show that there are three types of equilibria. In a Type I equilibrium, an exit by speculators is guaranteed along the equilibrium path because the speculators' optimal maximum bid exceeds the highest valuation of regular bidders. In a Type II equilibrium, a fork occurs in expectation, and the speculators' optimal bid falls within the range of regular bidders' valuations. Finally, in a Type III equilibrium, no fork occurs, and the speculators' optimal bid is zero because the expected redemption value is too low to participate profitably in the auction. We then proceed to explore four different redemption mechanisms and their implications for speculators' incentives. We further consider various extensions of our model, including atomic exits, time delays between exit and redemption, and spending paths of the organization, and provide numerical simulations of our main findings.

Our analysis adds to the literature on auctions with speculators by studying a setting in which speculators benefit not from a secondary market but by directly receiving their share of the organization's assets. Our results also contribute to understanding incentive structures and governance design in auction-based DAOs. The two ``innocent'' goals of (i) allowing unrestricted participation in governance by bidding for shares (and associated voting rights) and (ii) safeguarding members from majority attacks by allowing them to exit and redeem their share can combine in practice to yield unintended consequences---such as speculation and costly member exits.

\newpage
\appendix

\section{Additional Figures}
\begin{figure}
    \centering
    \includegraphics[width=0.4\linewidth]{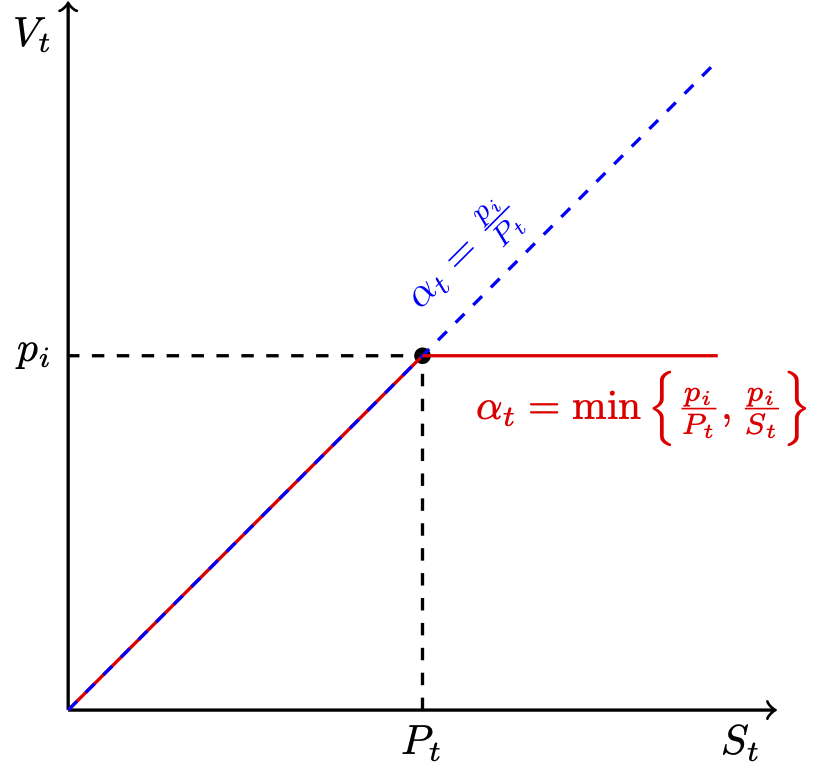}
    \caption{Contribution-based Share 1 and 2.}
    \label{fig:three}
\end{figure}

\section{Tables}

\begin{table}[htp]
    \centering
    \includegraphics[width=0.48\linewidth]{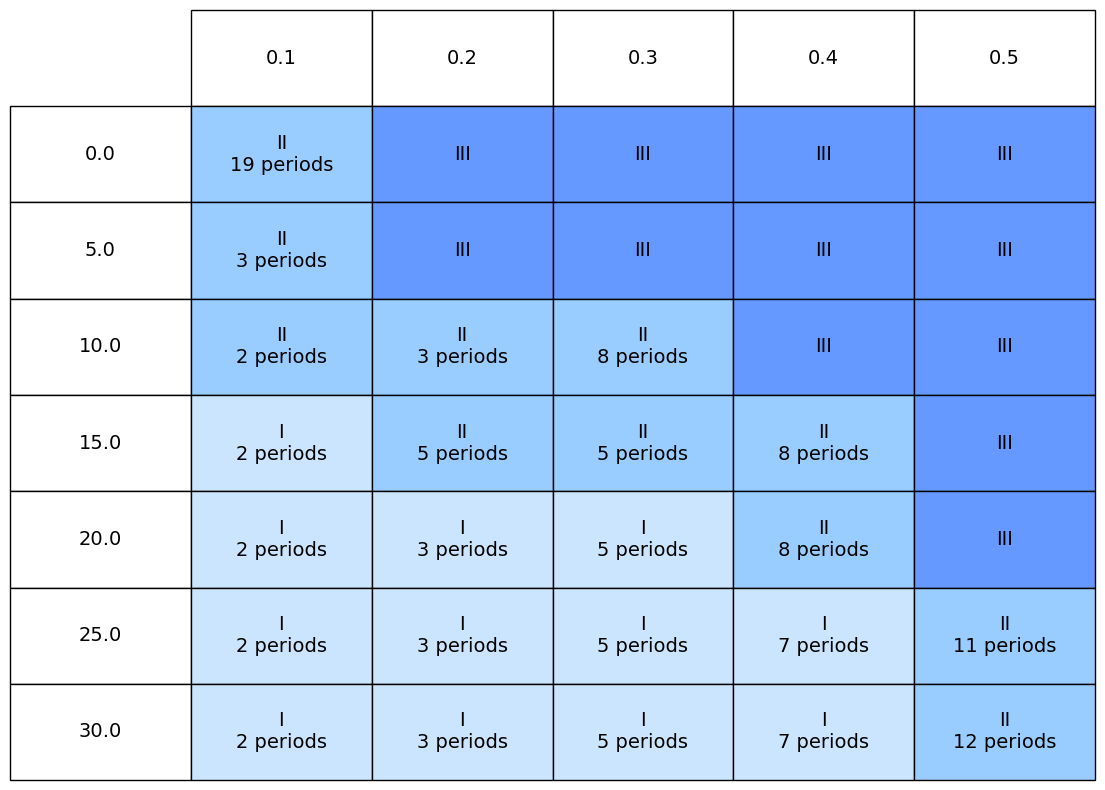}
    \vspace{15pt}
    \caption{Equilibrium types for different values of the initial treasury (vertical) and forking threshold (horizontal) and Mechanism 1.}
    \label{tab:one}
\end{table}
\begin{table}[htp]
    \centering
    \begin{subtable}{0.48\textwidth}
        \centering
        \includegraphics[width=\linewidth]{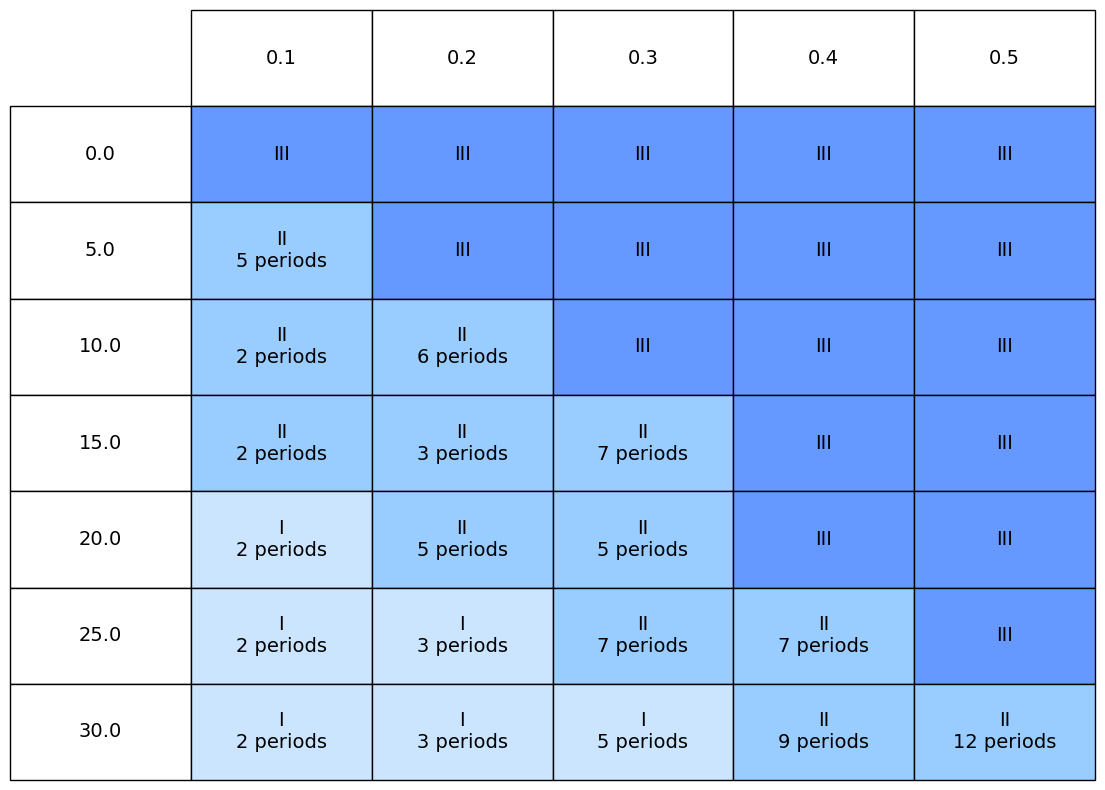}
        \caption{Tax of 25\%.}
        \label{tab:two}
    \end{subtable}
    \hfill
    \begin{subtable}{0.48\textwidth}
        \centering
        \includegraphics[width=\linewidth]{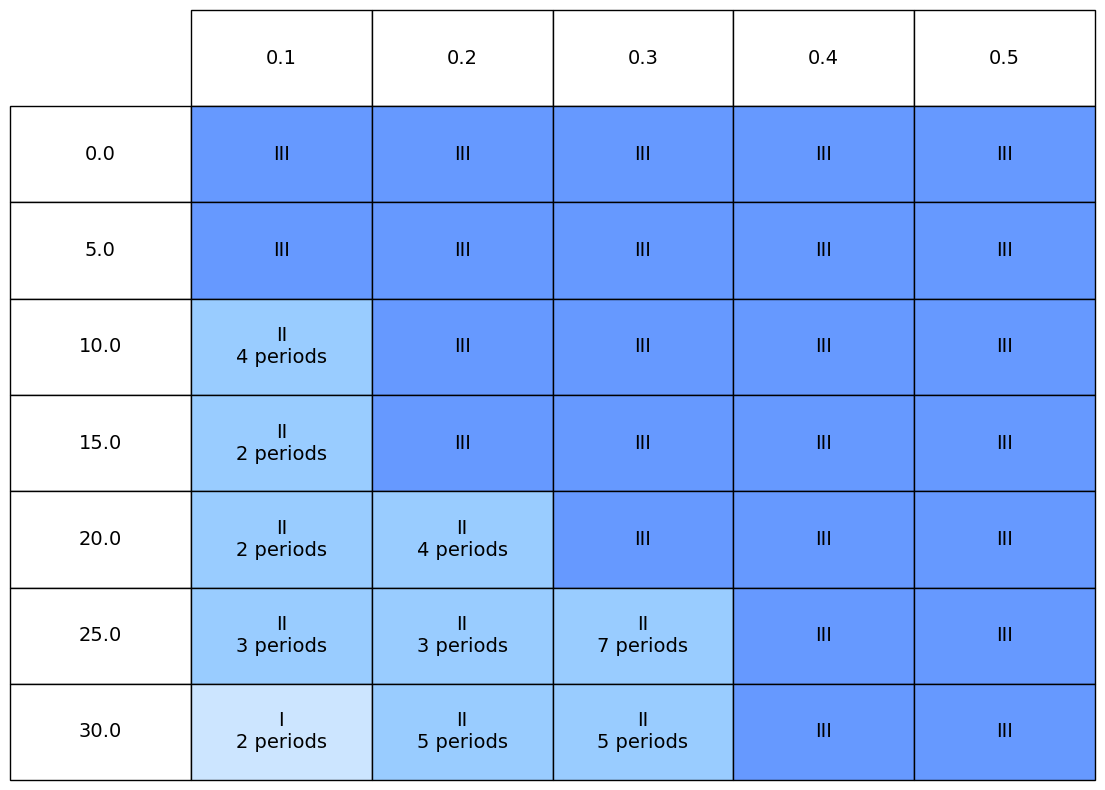}
        \caption{Tax of 50\%.}
        \label{tab:three}
    \end{subtable}
    \caption{Equilibrium types for different values of the initial treasury (vertical) and forking threshold (horizontal) and Mechanism 2.}
\end{table}
\begin{table}[htp]
    \centering
    \begin{subtable}{0.48\textwidth}
        \centering
        \includegraphics[width=\linewidth]{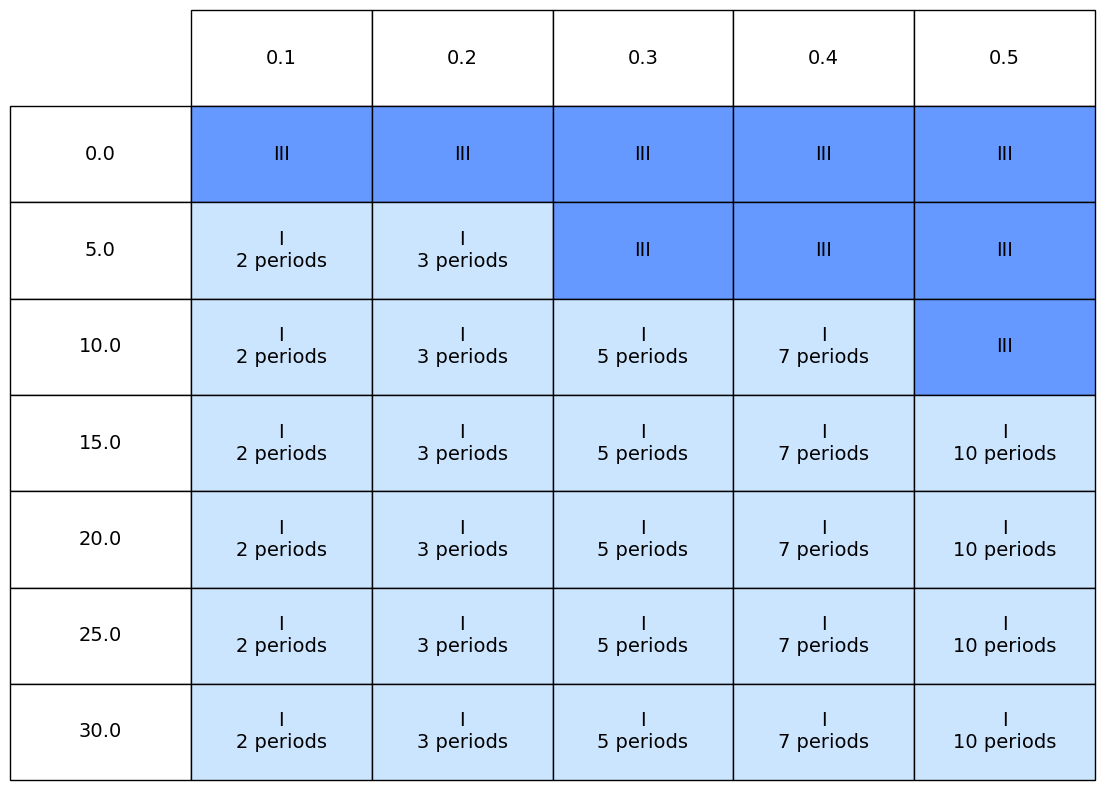}
        \caption{$P=0$}
        \label{tab:four}
    \end{subtable}
    \hfill
    \begin{subtable}{0.48\textwidth}
        \centering
        \includegraphics[width=\linewidth]{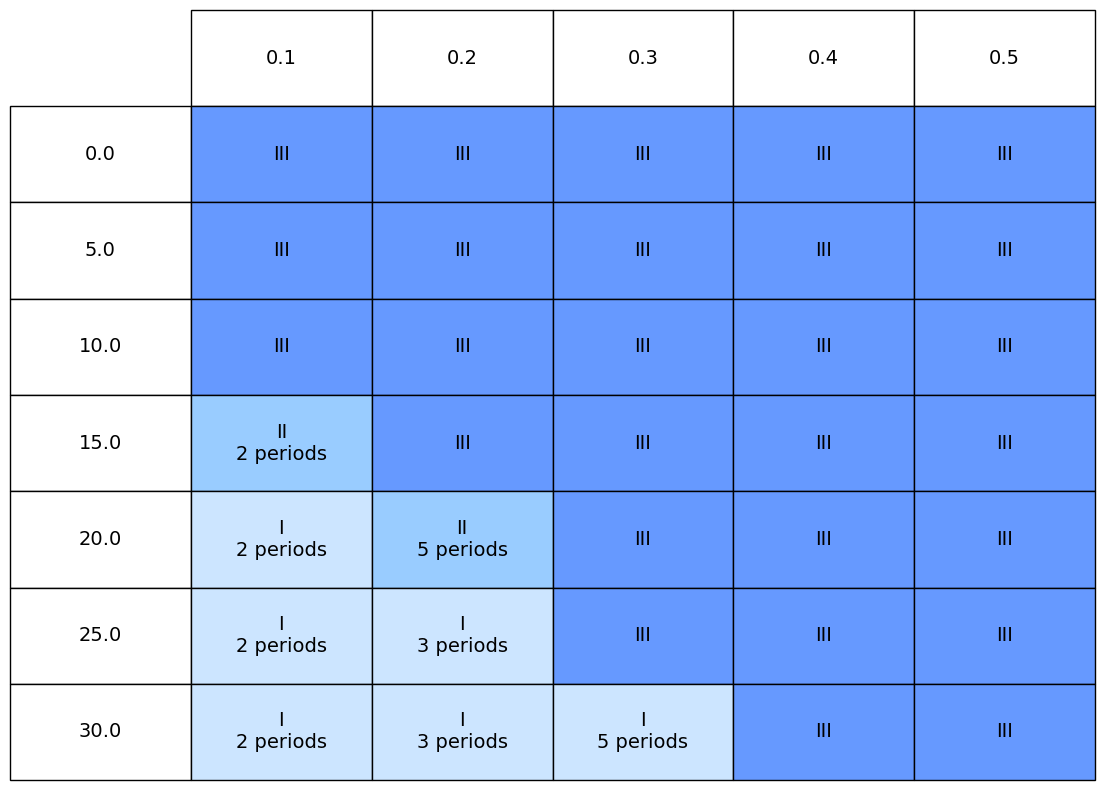}
        \caption{$P=10$}
        \label{tab:five}
    \end{subtable}
    \caption{Equilibrium types for different values of the initial treasury (vertical) and forking threshold (horizontal) and Mechanism 3.}
\end{table}

\newpage
\section{Proofs}
\subsection{Proof of Lemma 1:}

\begin{proof}
    We prove the statements in turn. 

\begin{enumerate}[i.]
    \item Let $d\mathcal{B}_{-i,t}$ denote the density of $\tilde{b}_{-i,t}$. The expected return for nouner $i$ of playing a strategy of bidding up to $\tilde{b}_{i,t}$ is given by
    \begin{equation}
        \int_{0}^{\tilde{b}_{i,t}} (v_{i,t}-\tilde{b}_{-i,t})d \mathcal{B}_{-i,t}.
    \end{equation}
    This expression is maximized at $\tilde{b}_{i,t} = v_{i,t}$ for any strategies of all players other than $i$ that give rise to the density $d \mathcal{B}_{-i,t}$.
    \item Fix a $V^e_t(h_t)$ and $\tilde{b}_{a,t}(h_t)$. First, the arbitrageur's utility function directly implies that holding the bid $\tilde{b}_{a,t}(h_t)$ constant, utility is increasing in $V^e_t(h^t)$, implying a maximum bid value $\hat{b}_t(h_t)$ exists and that $\tilde{b}_{a,t}(h_t) = \hat{b}_t(h_t)$ in equilibrium. Second, consider that an arbitrageur's bid has two effects on her payoff: changing the probability of winning and affecting the redemption value by changing today's auction price and future arbitrageur bids. Thus, if today's auction price is independent of an arbitrageur's bid, then so is $V^e_t(h_t)$. This is the case if $\hat{b}_t(h_t) > v_{n,t}$ or if $\hat{b}_t(h_t) \leq v_{n-1,t}$, where $v_{n,t}$ and $v_{n-1,t}$ are the highest and second-highest valuations of nouners at history $h_t$, respectively. By statement (i) and the definition of the arbitrageurs payoff, it then follows that $\hat{b} \geq \min\{v_{n-1,t}, V^e_t(h_t)\}$. \qed
\end{enumerate}
\end{proof}

\subsection{Proof of Lemma 2:}
\begin{proof}
Follows directly from the definition of the forking threshold $\kappa$ and the fact that by Lemma \ref{lem:bidding_strat} the expected number of arbitrageur wins is determined by the sum over all periods $\tau \in t, ..., t^\ast$ of the respective probability that $v_{n,\tau} \leq \hat{b}_{\tau}(h_\tau)$.\qed
\end{proof}

\subsection{Proof of Lemma 3:}

\begin{proof}
Consider an arbitrageurs' payoff, which can be rewritten by Lemma \ref{lem:bidding_strat} as follows
\begin{subeqnarray}
u_{a,t}(\tilde{b}_{a,t},h_t) &=& Prob[\tilde{b}_{a,t} > v_{n,t} | h_t] (V^e_t(h_t) - E[v_{n,t}]) \\
&=& \int_{0}^{\hat{b}_{a,t}} \left( V^e_t(h_t) - v \right) d\mathcal{B}_{-a,t} \\
&=& \int_{0}^{\hat{b}_{a,t}} \left( V^e_t(h_t) - v \right) n F(v)^{n-1} f(v) \, dv 
,
\end{subeqnarray} 
where $v_{n,t}$ denotes the highest drawn value among all nouners $i \in 1,...,n$ at history $h_t$ and $d\mathcal{B}_{-a,t}$ the density of the highest nouner valuation $v_{n,t}$.
Differentiating with respect to the arbitrageurs maximum bid $\hat{b}_{a,t}(h_t)$, we obtain
\begin{subeqnarray} 
\left( V^e_t(h_t) - \hat{b}_{a,t} \right) n F(\hat{b}_{a,t})^{n-1} f(\hat{b}_{a,t}) \quad \quad \quad \quad \quad \quad \quad \quad \quad & \\ + \ n \frac{\partial V^e_t(h_t)}{\partial \hat{b}_{a,t}} \int_{0}^{\hat{b}_{a,t}} n F(v)^{n-1} f(v) \, dv &= 0 \\
\left( V^e_t(h_t) - \hat{b}_{a,t} \right) n F(\hat{b}_{a,t})^{n-1} f(\hat{b}_{a,t}) +  \frac{\partial V^e_t(h_t)}{\partial \hat{b}_{a,t}} \frac{1}{n}F(\hat{b}_{a,t})^n &= 0 \\
V^e_t(h_t) -\hat{b}_{a,t} + \frac{\frac{\partial V^e_t(h_t)}{\partial \hat{b}_{a,t}} F(\hat{b}_{a,t})^n}{n F(\hat{b}_{a,t})^{n-1} f(\hat{b}_{a,t})} &=  0
\end{subeqnarray} which solves for
\begin{equation}
b^\ast_{t} = V^e_t(h_t) + \frac{1}{n} \frac{ F(b^\ast_{t})}{f(b^\ast_{t})} \frac{\partial V^e_t(h_t)}{\partial b^\ast_{t}}. 
\end{equation} \qed
\end{proof}

\subsection{Proof to Proposition 1:}
\begin{proof}
Fix an equilibrium. Let $b_t^\ast (h_t)$ denote the optimal maximum bid of an arbitrageur at history $h_t$ in equilibrium. Note first that arbitrageur strategies are symmetric. At any history $h_t$ and $t < \tau \leq t^\ast$ we must therefore have that $\mathbb{E}_t[b^\ast_\tau(h_\tau)] = b_t^\ast(h_t) \times g(\alpha,\tau, h_t)$, where $g(\alpha,\tau, h_t)$ is a weakly positive, continuous function. Formally, we can derive $g(\alpha,\tau, h_t)$ by considering the ratio
$$
\frac{b_t^\ast (h_t)}{\mathbb{E}[b_{\tau}^\ast]} = \frac{V^e_t(h_t) + \frac{1}{n} \frac{F(b_t^\ast(h_t))}{f(b_t^\ast(h_t))} \frac{\partial V^e_{t}(h_t)}{\partial b_t^\ast(h_t)}} {V^e_{\tau}(h_t) + \frac{1}{n} \frac{F(\mathbb{E}[b_{\tau}^\ast])}{f(\mathbb{E}[b_{\tau}^\ast])} \frac{\partial V^e_{\tau}(h_t)}{\partial \mathbb{E}[b_{\tau}^\ast]}}.
$$ We start with the redemption value $V^e_t(h_t)$, or
$$
V^e_t(h_t) = \delta^{t^\ast-t} \alpha_{t^\ast} \left( S_{t-1} + \sum_{\phi = t}^{t^\ast} \mathbb{E}[p_\phi] \right),
$$ where
\begin{subeqnarray}
\mathbb{E}[p_{\phi}] &=& Prob[v_{(n-1)} > b_\phi^\ast] \cdot \mathbb{E}[v_{(n-1)} \vert v_{(n-1)}\geq b_\phi^\ast]\\
&+& Prob[v_{(n-1)} \leq b_\phi^\ast < v_{(n)}] \cdot b_\phi^\ast\\
&+& Prob[v_{(n)} \leq b_\phi^\ast] \cdot \mathbb{E}[v_{(n)} \vert v_{(n)}\leq b_\phi^\ast].
\end{subeqnarray} The probabilities are given by
\begin{subeqnarray}
Prob[v_{(n-1)} > b_\phi^\ast] &=& 1 - n[F(b_\phi^\ast)]^{n-1}[1 - F(b_\phi^\ast)] - [F(b_\phi^\ast)]^n,\\
Prob[v_{(n-1)} \leq b_\phi^\ast < v_{(n)}] &=& n[F(b_\phi^\ast)]^{n-1}[1 - F(b_\phi^\ast)], \\
Prob[v_{(n)} \leq b_\phi^\ast] &=& [F(b_\phi^\ast)]^n,
\end{subeqnarray} and the expected values of the second-highest and highest bids are
\begin{subeqnarray}
\mathbb{E}[v_{(n-1)} \vert v_{(n-1)}\geq b_\phi^\ast] &=& \frac{\int_{b_{\phi}^\ast}^{\bar{v}} v \cdot n (n-1) [F(v)]^{n-2} [1 - F(v)] f(v) \, dv}{1 - n[F(b_\phi^\ast)]^{n-1}[1 - F(b_\phi^\ast)] - [F(b_\phi^\ast)]^n}, \\
\mathbb{E}[v_{(n)} \vert v_{(n)}\leq b_\phi^\ast] &=& \frac{\int_0^{b_\phi} v \cdot n [F(v)]^{n-1} f(v) \, dv}{[F(b_\phi^\ast)]^n},
\end{subeqnarray} so that we obtain
\begin{subeqnarray}
V^e_t(h_t) &=& \delta^{t^\ast-t} \alpha_{t} \biggl( S_{t-1} + \sum_{\phi=t}^{t^\ast} \biggl[ \mathbb{E}[b_\phi^\ast] \cdot nF\left(\mathbb{E}[b_\phi^\ast] \right)^{n-1} \left(1 - F\left(\mathbb{E}[b_\phi^\ast] \right)\right) \biggr.  \\
&+& \int_{b_\phi^\ast}^{\bar{v}} v \cdot n(n-1)F(v)^{n-2}(1 - F(v)) f(v) \, dv \\
&+& \biggl. \int_0^{b_\phi^\ast} v \cdot nF(v)^{n-1} f(v) \, dv \biggr] \biggr).
\end{subeqnarray}
Differentiating the price $\mathbb{E}[p_\phi]$ with respect to $\mathbb{E}[b_\phi^\ast]$, we find
\begin{subeqnarray}
&& \frac{\partial \mathbb{E}[p_\phi]}{\partial \mathbb{E}[b_\phi^\ast]} \left( \mathbb{E}[b_\phi^\ast] \cdot n \cdot F(\mathbb{E}[b_\phi^\ast])^{n-1} \cdot (1 - F(\mathbb{E}[b_\phi^\ast])) \right)\\ 
&+& \frac{\partial}{\partial \mathbb{E}[b_\phi^\ast]} \int_0^{\mathbb{E}[b_\phi^\ast]} v \cdot nF(v)^{n-1} f(v) \, dv \\
&+& \frac{\partial}{\partial \mathbb{E}[b_\phi^\ast]} \int_{\mathbb{E}[b_\phi^\ast]}^{\bar{v}} v \cdot n(n-1)F(v)^{n-2}(1 - F(v)) f(v) \, dv \\[1em]
&=& n \cdot F(\mathbb{E}[b_\phi^\ast])^{n-1} \cdot (1 - F(\mathbb{E}[b_\phi^\ast])) \\ 
&+& \mathbb{E}[b_\phi^\ast] \cdot n \cdot f(\mathbb{E}[b_\phi^\ast]) \left( (n-1) F(\mathbb{E}[b_\phi^\ast])^{n-2} (1 - F(\mathbb{E}[b_\phi^\ast])) - F(b_t^\ast)^{n-1} \right)\\
&+& \mathbb{E}[b_\phi^\ast] \cdot nF(\mathbb{E}[b_\phi^\ast])^{n-1} f(\mathbb{E}[b_\phi^\ast])\\
&-& \mathbb{E}[b_\phi^\ast] \cdot n(n-1)F(\mathbb{E}[b_\phi^\ast])^{n-2}(1 - F(\mathbb{E}[b_\phi^\ast])) f(\mathbb{E}[b_\phi^\ast]) \\[1em]
&=& n \cdot F(\mathbb{E}[b_\phi^\ast])^{n-1} \cdot (1 - F(\mathbb{E}[b_\phi^\ast])) \\[1em]
&\geq& 0,
\end{subeqnarray} where the final inequality follows from the definition of the CDF $F(\cdot)$, which satisfies $0 \leq F(b_\phi^\ast) \leq 1$. Now, finally substituting into the ratio we are interested in, we obtain
\begin{subeqnarray}
&& \tfrac{b_t^\ast (h_t)}{\mathbb{E}[b_{\tau}^\ast]}\\ &=&  \tfrac{\delta^{t^* - t} \alpha_t \left( S_t + \sum_{\phi = t}^{t^*} \mathbb{E}[p_\phi] \right) + \tfrac{1}{n} \tfrac{F(b_t^*)}{f(b_t^*)} \delta^{t^* - t} \alpha_t \sum_{\phi = t}^{t^*} \left( n F(b_\phi^*)^{n-1} (1 - F(b_\phi^*)) \right)}{\delta^{t^* - \tau} \alpha_{\phi} \left( S_t + \sum_{\tau = \tau}^{t^*} \mathbb{E}[p_\tau] \right) + \tfrac{1}{n} \tfrac{F(\mathbb{E}[b_{\tau}^*])}{f(\mathbb{E}[b_{\tau}^*])} \delta^{t^* - \tau} \alpha_{\tau} \sum_{\tau = \tau}^{t^*} \left( n F(\mathbb{E}[b_{\tau}^\ast])^{n-1} (1 - F(\mathbb{E}[b_{\tau}^\ast])) \right)} \\[1em]
&=& \tfrac{\delta^{t^\ast-t}}{\delta^{t^\ast-\tau}} \tfrac{\alpha_t}{\alpha_\tau} \left[ \tfrac{\left( S_t + \sum_{\tau = t}^{t^*} \mathbb{E}[p_\tau] \right) + \tfrac{1}{n} \tfrac{F(b_t^*)}{f(b_t^*)} \sum_{\tau = t}^{t^*} \left( n F(b_t^*)^{n-1} (1 - F(b_t^*)) \right)}{ \left( S_t + \sum_{\tau = \tau}^{t^*} \mathbb{E}[p_\tau] \right) + \tfrac{1}{n} \tfrac{F(\mathbb{E}[b_{\tau}^*])}{f(\mathbb{E}[b_{\tau}^*])} \sum_{\tau = \tau}^{t^*} \left( n F(\mathbb{E}[b_{\tau}^\ast])^{n-1} (1 - F(\mathbb{E}[b_{\tau}^\ast])) \right)} \right] \\
&=& \frac{\delta^{t^\ast-t}}{\delta^{t^\ast-\tau}} \frac{\alpha_t}{\alpha_\tau} \times a_{t,\tau} \\
&=&  \frac{1}{g(\alpha, \tau, h_t)} \\[1em]
&\geq& 0
\end{subeqnarray} where the final inequality follows from $a_{t,\tau} \geq 0$, where $a_{t,\tau}$ is defined as the term in square brackets and all of its elements are weakly positive.

Now consider the necessary condition for a fork to occur in Lemma \ref{lem:forking}. Replacing the cutoff bids with the increase of expected optimal bids by arbitrageurs, we obtain that the expected time to fork at $h_t$ is given by
\begin{equation}
\sum_{\tau = t}^{t^\ast} \left[F\left(b_1^\ast \cdot \frac{1}{\delta^{\tau-t}}\frac{\alpha_\tau}{\alpha_t} \frac{1}{a_{t,\tau}}\right)\right]^n \geq \kappa (N + t^\ast) - A_t,
\end{equation} which must be satisfied for a fork to occur in expectation at $h_1$ and thus at any $h_t$. Moreover, as $\partial V_t^e(h_t) / \partial \mathbb{E}[p_\tau] \geq 0$, we find that
\begin{subeqnarray}
    \frac{\partial V_t^e(h_t)}{\partial \mathbb{E}[b_\tau^\ast]} &=& \delta^{t^\ast-t} \alpha_{t} \sum_{\tau=t}^{t^\ast} \biggl( n F\left(\frac{b_t^\ast}{\delta^{\tau-t}} \frac{\alpha_\tau}{\alpha_t} \frac{1}{a_{t,\tau}} \right)^{n-1} \biggr. \\ &\cdot& \biggl. \left(1 - F\left(\frac{b_t^\ast}{\delta^{\tau-t}} \frac{\alpha_\tau}{\alpha_t} \frac{1}{a_{t,\tau}} \right) \right) \biggr)\\ &\geq& 0.
\end{subeqnarray}

By Lemma 3 we therefore find that if $b_t^\ast \in [0,\bar{v}]$, then $b_t^\ast \geq V^e_t(h_t)$.  In conjunction, this yields the conditions for Type II and Type III in Proposition~\ref{prop:equilibrium} and the sufficient condition for Type II.

Lastly, consider the case when $b_1^\ast \notin (0,\bar{v}]$. First, we know by Lemma \ref{lem:arbitrageur_bid} and the fact that $\partial V^e_t(h_t)/\partial \mathbb{E}[p_\tau] > 0$ that if $V^e_t(h_t) = \bar{v}$, then $b_1^\ast \geq \bar{v}$. We thus find that if 
\begin{equation}
V^e_t(h_t) = \delta^{T-t} \alpha_t S_0 \geq \bar{v} \ \Rightarrow \ S_0 \geq \frac{\bar{v}}{\delta^{T-t} \alpha_t}
\end{equation} and
\begin{equation}
T \geq \kappa (N+T) \ \Rightarrow \ T \geq \frac{\kappa}{1-\kappa}N
\end{equation} then $b_t^\ast(h_t) \geq \bar{v}$. This completes the proof for Type I.
\qed

\end{proof}

\subsection{Proof of Corollary 1:}
\begin{proof}
Consider the redemption value, which under Mechanism \ref{mech:four} and given a bid at time $t$ and history $h_t$ becomes
\begin{subeqnarray}
    V^e_t(h_t) &=& \delta^{t^\ast -t} \alpha_{j,t}(h_t) \left( S_0 +  \sum_{\tau=t}^{t^\ast} \mathbb{E}\left[p_\tau \right] \right) \\
&=& \delta^{t^\ast - t} \frac{p_{j=t}}{S_0 + \sum_{\tau=t}^{t^\ast} \mathbb{E}\left[p_\tau \right]} \left( S_0 +  \sum_{\tau=t}^{t^\ast} \mathbb{E}\left[p_\tau \right] \right) \\
&=&p_{j=t} \delta^{t^\ast - t}.
\end{subeqnarray}
where without loss of generality we set $\alpha_{j,t}(h_t) = p_j/S_t$, as $S_0 \nless 0$. An arbitrageur's utility, therefore, becomes
\begin{subeqnarray}
    u_{a,t}(b^\ast_t(h_t),h_t) &=& Prob[b^\ast_t(h_t) > v_{n,t}] (V^e_t(h_t) - p_{j=t}) \\
&=& Prob[b^\ast_t(h_t) > v_{n,t}] (p_{j=t}(\delta - 1) ) \\
&<& 0 \ \ \forall \ b_t^\ast(h_t) > 0,
\end{subeqnarray}
implying that in equilibrium $b_t^\ast(h_t) = 0$.
\qed
\end{proof}

\subsection{Proof of Proposition 2:}

\begin{proof}
Assume $\kappa = 0$. Consider the arbitrageur's expected payoff, which then simplifies to 
\begin{subeqnarray}
    u_{a,t}(b^\ast,h_t) &= \int_{0}^{b^\ast} \left( V(h_t) - v \right) n F(v)^{n-1} f(v) \, dv,
\end{subeqnarray}
and solves for
\begin{subeqnarray}
b^\ast(h_t) &= V(h_t) + \frac{1}{n} \frac{ F(b^\ast(h_t))}{f(b^\ast(h_t))} \frac{\partial V(h_t)}{\partial b^\ast(h_t)}.
\end{subeqnarray} We know by the proof of Proposition \ref{prop:equilibrium} that 
 $b^\ast(h_t) \geq V(h_t)$. 
We thus find that in expectation, the arbitrageur will win the auction if 
\begin{equation}
    V(h_t) > v_{n,t}. 
\end{equation}
Because a fork occurs each time an arbitrageur wins, the treasury will shrink at time $t$ and history $h_t$ in expectation if
\begin{equation}
    \alpha_t F(b^\ast(h_t))^n (S_{t-1} + \mathbb{E}[p_t]) > \mathbb{E}[p_t],
\end{equation}
or
\begin{equation}
S_{t-1} >\left( \frac{1}{\alpha_t F(b^\ast(h_t))^n} -1 \right) \mathbb{E}[p_t].
\end{equation} Now note that as $S_{t-1} \rightarrow 0$, this condition becomes impossible to fulfill, implying that the treasury must increase. Similarly observe that as $S_{t-1} \rightarrow 1/\alpha_t \times \bar{v}$, the condition is guaranteed to be fulfilled because $p_t \leq \bar{v}$ and $F(b^\ast(h_t)) \rightarrow 1$ as $S_{t-1} \rightarrow 1/\alpha_t \times \bar{v}$. Finally, note that we know from the proof of Proposition \ref{prop:equilibrium} that $\partial \mathbb{E}[p_t]/\partial b^\ast(h_t)$ is a continuous function which, in conjunction with the definition of $F(v)$ implies that there is a single solution for a bid that leaves the treasury unchanged in expectation. \qed
\end{proof}

\subsection{Proof of Corollary 2:}

\begin{proof}
We prove the statements in turn. Consider that under Mechanism \ref{mech:three}, an arbitrageur has a certain positive profit if he bids $b_t = \bar{v}$ if $S_{t-1}>P_{t-1}$, as 
\begin{equation}
    \frac{\mathbb{E}[p_t]}{P_{t-1} +\mathbb{E}[p_t]} (S_{t-1} +  \mathbb{E}[p_t]) -  \mathbb{E}[p_t] > 0 \quad \forall \ S_{t-1}>P_{t-1}
\end{equation}
where $\mathbb{E}[p_t] = \mathbb{E}[p_t \vert b_t = \bar{v}]$. Under Mechanism \ref{mech:four}, in turn, in any equilibrium, the arbitrageur plays $b_t^\ast(h_t) = 0$, as
\begin{subeqnarray}
F(b_t)^n && \min \left\{ \frac{\mathbb{E}[p_t]}{S_{t-1} + \mathbb{E}[p_t]}, \frac{\mathbb{E}[p_t]}{P_{t-1} + \mathbb{E}[p_t]} \right\}  (S_t + \mathbb{E}[p_t]) - \mathbb{E}[p_t] \\
& \leq & \min \left\{ \frac{\mathbb{E}[p_t]}{S_{t-1} + \mathbb{E}[p_t]}, \frac{\mathbb{E}[p_t]}{P_{t-1} + \mathbb{E}[p_t]} \right\}  (S_t + \mathbb{E}[p_t]) - \mathbb{E}[p_t] \\
& \leq & \frac{\mathbb{E}[p_t]}{S_{t-1} + \mathbb{E}[p_t]} (S_t + \mathbb{E}[p_t]) - \mathbb{E}[p_t] \\
&=& 0.
\end{subeqnarray} where $\mathbb{E}[p_t] = \mathbb{E}[p_t \vert b_t]$ for any $b_t \in [0,\bar{v}]$. \qed
\end{proof}

\subsection{Proof of Proposition 3:}

\begin{proof}
Consider the forking condition
\begin{equation}
\sum_{\tau = t}^{t^\ast} \left[F\left(b_1^\ast \cdot \frac{1}{\delta^{\tau-t}}\right)\right]^n \geq \kappa (N + t^\ast) - A_t,
\end{equation} optimal bid
\begin{equation}
b_t^\ast(h_t) = V^e_t(h_t) + \frac{1}{n} \frac{ F(b_t^\ast(h_t))}{f(b_t^\ast(h_t))} \frac{\partial V^e_t(h_t)}{\partial b_t^\ast(h_t)},
\end{equation} and redemption value under the committed spending path
\begin{equation}
V^e_t(h_t) = \delta^{t^\ast-t} \alpha_{t^\ast} \left( S_{t-1} + \sum_{\tau = t}^{t^\ast} \left( \mathbb{E}[p_\tau] -  k S_{\tau-1} e^{-\lambda(\tau-1)} \right) \right).
\end{equation} Observe that 
\begin{equation}\frac{\partial V^e_t(h_t)}{\partial z_t(S_{t-1})} < 0, \qquad \frac{\partial b^\ast_t(h_t)}{\partial V^e_t(h_t)} > 0, \qquad \frac{\partial 
\sum_{\tau = t}^{t^\ast} \left[F\left(b_1^\ast \cdot \frac{1}{\delta^{\tau-t}}\right)\right]^n}{\partial b_1^\ast} > 0,
\end{equation} which immediately delivers the result. \qed

\end{proof}

\end{document}